\documentclass[11pt]{article}
\usepackage{epsfig,psfrag}
\usepackage{times}
\usepackage{amsmath}
\usepackage{amssymb}
\usepackage{bbm}
\usepackage{amscd}
\usepackage{xspace,color}
\input xy
\xyoption{all}

\newcommand{\eop}{\hfill\usebox{\smallProofsym}\bigskip}  %
\newsavebox{\smallProofsym}                            
\savebox{\smallProofsym}                               %
{                                                     %
\begin{picture}(6,6)                                   %
\put(0,0){\framebox(6,6){}}                            %
\put(0,2){\framebox(4,4){}}                            %
\end{picture}                                          %
}                                                      %

\makeatletter
\long\def\@makecaption#1#2{%
  \vskip\abovecaptionskip
  \sbox\@tempboxa{\small #1: #2}%
  \ifdim \wd\@tempboxa >\hsize
    \small #1: #2\par
  \else
    \global \@minipagefalse
    \hb@xt@\hsize{\hfil\box\@tempboxa\hfil}%
  \fi
  \vskip\belowcaptionskip}
\makeatother

\newcommand {\mm}[1] {\ifmmode{#1}\else{\mbox{\(#1\)}}\fi}

\newcommand{\eps}           {\mm{\varepsilon}}

\newtheorem{theorem}{Theorem}

\newtheorem{proposition}[theorem]{Proposition}

\newtheorem{claim}[theorem]{Claim}

\title{Normal variation for adaptive feature size}
\author{Nina Amenta\thanks{Dept. of CS, U. of California, Davis, CA 95616,
amenta@cs.ucdavis.edu}~~ and~~ Tamal K. Dey\thanks{Dept. of CSE, Ohio State U.,
Columbus, OH 43210, tamaldey@cse.ohio-state.edu}}

\begin{document}
\maketitle

\section*{Background}
Let $\Sigma$ be a closed, smooth surface in $\mathbb{R}^3$.
For any two sets $X,Y\subset\mathbb{R}^3$, let $d(X,Y)$
denote the Euclidean distance between $X$ and $Y$.
The local feature size $f(x)$ at a point $x\in\Sigma$ 
is defined to be the distance $d(x,M)$ where $M$ is
the medial axis of $\Sigma$. 
Let $n_p$ denote the unit normal (inward) to $\Sigma$ at point $p$.
Amenta and Bern in their paper~\cite{AB98} claimed
the following:
\begin{claim}
Let $q$ and $q'$ be any two points in $\Sigma$
so that $d(q,q')\leq \eps \min\{f(q),f(q')\}$
for $\eps\leq \frac{1}{3}$.
Then, $\angle{n_q,n_{q'}}\leq \frac{\eps}{1-3\eps}$. 
\end{claim}

Unfortunately, the proof of this claim as given in
Amenta and Bern~\cite{AB98} is wrong; it  also
appears in the book by Dey~\cite{Dey06}.  
In this short note, we provide a correct
proof with an improved bound of $\frac{\eps}{1-\eps}$.

\begin{theorem}
Let $q$ and $q'$ be two points in $\Sigma$
with $d(q,q')\leq \eps f(q)$ where $\eps \leq \frac{1}{3}$.
Then, $\angle{n_q,n_{q'}} \leq \frac{\eps}{1-\eps}$.
\label{nor-thm}
\end{theorem}

\section*{Definitions and Preliminaries}
For any point $p\in\mathbb{R}^3$, let $\tilde p$ denote the
closest point of $p$ in $\Sigma$. When $p$ is a point
in $\Sigma$, the normal to $\Sigma$ at $p$ is well defined.
We extend this definition to any point $p\in\mathbb{R}^3$.
Define the normal $n_p$ at $p\in\mathbb{R}^3\setminus M$ as the 
normal to $\Sigma$ at $\tilde p$. Similarly, we extend the
definition of local feature size $f$ to $\mathbb{R}^3$.
For any point $p\in\mathbb{R}^3$, let $f(p)$ be the distance
of $p$ to the medial axis of $\Sigma$. Notice that $f$ is
$1$-Lipschitz.
If two points $x$ and $y$ lie on a surface $F\subset \mathbb{R}^3$, let  
$d_{F}(x,y)$ denote the {\em geodesic distance} 
between $x$ and $y$.
The following facts are well known in differential geometry.

\begin{proposition}
\label{propnor}
Let $F$ be a smooth surface in $\mathbb{R}^3$. Let $q$ and
$q'$ be two points in $F$. Then, 
$$
\lim_{d\rightarrow 0}\frac{d_{F}(q,q')}{d(q,q')}=1.
$$
\end{proposition}

\begin{proposition}
\label{propkappa}
Consider the geodesic path between $q,q'$ on a smooth surface
$F$ in $\mathbb{R}^3$. 
Let $\kappa_{m}$
be the maximum curvature on this geodesic path.
Then 
$\angle{n_q,n_{q'}} \leq \kappa_{m} d_{F}(q,q')$. 
\end{proposition}

\section*{The Proof}
We are to measure
$\angle{n_q,n_{q'}}$ for two points $q$ and $q'$
in $\Sigma$. 
One approach would be to use the propositions above to bound the length 
of a path from $p$ to $q$ on $\Sigma$ and then use that length to bound the
change in normal direction, 
but we can get a better bound by 
considering the direct path from $p$ to $q$. 

Let $\Sigma_{\omega}$ denote an 
offset of $\Sigma$, that is, each point
in $\Sigma_{\omega}$ has distance $\omega$ from $\Sigma$. 
Formally, consider the distance function
$$
h \colon \mathbb{R}^3\rightarrow \mathbb{R},\,\, 
h(x)\mapsto d(x,\Sigma).
$$
Then, $\Sigma_{\omega}=h^{-1}({\omega})$.

\begin{claim}
For $\omega\geq 0$
let $p$ be a point in $\Sigma_{\omega}$ where $\omega<f(\tilde p)$.
There is an open set $U\subset \mathbb{R}^3$ so that
$\sigma_p=\Sigma_{\omega}\cap U$ is a smooth $2$-manifold
which can be oriented so that
$n_x$ is the normal to $\sigma_p$ at any $x\in\sigma_p$.
\label{nor-claim}
\end{claim}
\begin{proof}
Since $\omega<f(\tilde p)$, $p$ is not a point on the medial
axis. Therefore, the distance function $h$
is smooth at $p$. One can apply the implicit function theorem
to claim that there exists an open set $U\subset\mathbb{R}^3$
where
$$
\sigma_p=h^{-1}(\omega)\cap U
$$
is a smooth $2$-manifold. 
The unit gradient 
$(\frac{\nabla h}{\|\nabla h\|})_x=\frac{x-\tilde x}{\|x-\tilde x\|}$
which is precisely $n_x$ up to orientation is
normal to $\sigma_p$ at $x\in\sigma_p$. 
\end{proof}$\eop$

\begin{proof}[{\sf Proof of Theorem~\ref{nor-thm}}]
Consider parameterizing the segment $qq'$ by the length of $qq'$. 
Take two arbitrarily close points $p=p(t)$ and $p'=p(t+\Delta t)$
in $qq'$ for arbitrarily small $\Delta t>0$. 
Let $\theta(t)=\angle{n_q,n_{p(t)}}$ 
and $\Delta \alpha=\angle{n_p,n_{p'}}$. Then,
$|\theta(t+\Delta t)-\theta(t)|\leq \Delta \alpha$
giving 
$$|\theta'(t)|\leq \lim_{\Delta t\rightarrow 0} \frac{\Delta\alpha}{\Delta t}.
$$ 
If we show that $\lim_{\Delta t\rightarrow 0} \frac{\Delta\alpha}{\Delta t}$
is no more than $\frac{1}{(1-\eps)f(q)}$
we are done since then
\begin{eqnarray*}
\angle{n_q,n_{q'}}&\leq& \int_{qq'}  |\theta'(t)|dt\\
                  &\leq& \int_{qq'}\frac{1}{(1-\eps)f(q)} dt\\
                  &=& \frac{d(q,q')}{(1-\eps)f(q)}\\
		  &\leq& \frac{\eps}{(1-\eps)}.
\end{eqnarray*} 

We have $d(q,\tilde{p}) \leq d(q,p) + d(p,\tilde{p})$ and $d(q,p) \leq \eps f(q)$.
Since also $\omega = d(p,\tilde{p}) \leq d(p,q) \leq \eps f(q)$,
we have $\omega \leq \frac{2\eps}{1-2 \eps}f(\tilde p)$
(by a standard argument using the fact that the function $f$ is $1$-Lipshitz).
Therefore, $\omega < f(\tilde p)$ for $\eps < 1/3$, and 
there is a smooth neighborhood $\sigma_p\subset \Sigma_{\omega}$ of $p$
satisfying Claim~\ref{nor-claim}. 

Let $r$ be the closest
point to $p'$ in $\Sigma_{\omega}$, and let 
$\Delta t$ be small enough so that  $r$ and
the geodesic between $p$ and $r$ in $\sigma_p$ lies in $\sigma_p$. 
Notice that , by Claim~\ref{nor-claim}, 
$\Delta \alpha = \angle{n_p,n_{p'}} = \angle{n_p,n_{r}}$. 
\begin{claim}
$\lim_{\Delta t\rightarrow 0} \frac{d(p,r)}{\Delta t}\leq 1.$
\label{close}
\end{claim}
\begin{proof}
Consider the triangle $prp'$. If the tangent plane to
$\sigma_p$ at $r$ separates $p$ and $p'$, the
angle $\angle{prp'}$ is obtuse. It follows that
$d(p,r)\leq d(p,p')=\Delta t$. In the other case when
the tangent plane to $\sigma_p$ at $r$ does not
separate $p$ and $p'$, the angle $\angle{prp'}$ is
non-obtuse. Let $x$ be the foot of the perpendicular
dropped from $p$ on the line of $p'r$. We have
$d(p,r)\cos\alpha \leq d(p,p')$ where $\alpha$ is
the acute angle $\angle{rpx}$. 
Combining the two cases we have
$d(p,r)/\Delta t \leq \frac{1}{\cos\alpha}$.
Since $\alpha$ goes to $0$ as $\Delta t$ goes to $0$, we have
$\lim_{\Delta t\rightarrow 0} \frac{d(p,r)}{\Delta t} \leq 1$.
\end{proof}$\eop$ 

Now consider the geodesic between $p$ and $r$ in $\sigma_p$,  
and let $m$ be the point on the geodesic at which the maximum curvature 
$\kappa_{m}$ is realized. 
Recall that $d_{\sigma_p}(p,r)$ denotes the geodesic distance
between $p$ and $r$ on $\sigma_p$.
Let $r_m$ be the radius
of curvature corresponding to $\kappa_{m}$, i.e., $\kappa_{m}=1/r_m$.
Clearly, $f(m) \leq r_m$. So, Proposition~\ref{propkappa} tells us that 
$$
\Delta\alpha\leq \frac{d_{\sigma_p}(p,r)}{f(m)}.
$$

Therefore,
\begin{eqnarray*}
\lim_{\Delta t\rightarrow 0} \frac{\Delta\alpha}{\Delta t}
&\leq& \lim_{\Delta t\rightarrow 0} \frac{d_{\sigma_p}(p,r)}{\Delta t \: f(m)} \\
\end{eqnarray*}

In the limit when $\Delta t$ goes to zero, $d_{\sigma_p}(p,r)$
approaches $d(p,r)$ which in turn approaches $\Delta t$ (Proposition~\ref{propnor}
and Claim~\ref{close}). 
Meanwhile, $d(q,m) \leq d(q,p) + d(p,r)$ approaches $d(q,p)\leq \eps f(q)$ as
$\Delta t$ goes to zero (again by Claim~\ref{close}).
So, in the limit, $f(m) > (1-\eps) f(q)$ (again using the fact that $f$ is $1$-Lipshitz). 
Therefore,
$$
\lim_{\Delta t\rightarrow 0}\frac{\Delta\alpha}{\Delta t}\leq \frac{1}{(1-\eps)f(q)}
$$
which is what we need to prove.
\end{proof}$\eop$

\vspace{0.2in}
\noindent
{\sf Remark}: The bound on normal variation can be slightly improved
to $-\ln(1-\eps)$ by observing the following. We used that
$d(q,p)\leq \eps f(q)$ to arrive at the bound $f(m)>(1-\eps)f(q)$.
In fact, one can observe that $d(q,p)\leq \eps t f(q)$ giving
$f(m)>(1-\eps t)f(q)$. This gives $|\theta'(t)|\leq \frac{1}{(1-\eps t)f(q)}$.
We have
$$
\angle{n_q,n_{q'}}\leq \int_{qq'}\frac{1}{(1-\eps t)f(q)}dt
=- \frac{d(q,q')\ln(1-\eps)}{\eps f(q)}= -\ln(1-\eps).
$$

\vspace{0.2in}
\noindent
{\bf Acknowledgement} We thank Siu-Wing Cheng, Jian Sun, and Edgar
Ramos for discussions on the normal variation problem.

\end{document}